\documentclass[runningheads]{llncs}
\usepackage{amssymb}
\usepackage[usenames,dvipsnames,svgnames,table]{xcolor}
\usepackage{amsmath}
\usepackage{graphicx}
\usepackage{hyperref}
\usepackage{color}
\usepackage{epstopdf}
\usepackage{cleveref}
\usepackage[svgnames]{xcolor}
\usepackage{braket}
\usepackage{cite}
\hypersetup{hidelinks,colorlinks=true,allcolors=DarkBlue}

\usepackage[normalem]{ulem}

\newcommand{\sk}{\text{\rm \sffamily{sk}}}
\newcommand{\pk}{\text{\rm \sffamily{pk}}}
\newcommand{\kg}{\text{\rm \sffamily{Kg}}}
\newcommand{\sign}{\text{\rm \sffamily{Sign}}}
\newcommand{\vf}{\text{\rm \sffamily{Vf}}}

\newtheorem{mydef}{Definition}
\newtheorem{mylemma}{Lemma}
\newtheorem{mytheorem}{Theorem}
\newtheorem{myremark}{Remark}
\newcommand{\DSS}{S}
\newcommand{\LOTS}{{\rm L-OTS}}
\newcommand{\WOTS}{{\rm W-OTS}$^{+}$}

\begin{document}
\title{Proof-of-forgery for hash-based signatures}

\author
{
E.O. Kiktenko\and
M.A. Kudinov\and
A.A. Bulychev\and
A.K. Fedorov
}

\authorrunning{E.O. Kiktenko et al.}
\institute{Russian Quantum Center, Skolkovo, Moscow 143025, Russia\\
\email{e.kiktenko@rqc.ru; akf@rqc.ru}}

\maketitle
\begin{abstract}
In the present work, a peculiar property of hash-based signatures allowing detection of their forgery event is explored.
This property relies on the fact that a successful forgery of a hash-based signature most likely results in a collision with respect to the employed hash function, 
while the demonstration of this collision could serve as convincing evidence of the forgery.
Here we prove that with properly adjusted parameters Lamport and Winternitz one-time signatures schemes could exhibit a forgery detection availability property.
This property is of significant importance in the framework of crypto-agility paradigm since the considered forgery detection serves as an alarm that the employed cryptographic hash function becomes insecure to use and the corresponding scheme has to be replaced.
\end{abstract}

\section{\uppercase{Introduction and problem statement}}
\label{sec:introduction}

Today, cryptography is an essential tool for protecting the information of various kinds.
A particular task that is important for modern society is to verify the authenticity of messages and documents effectively. 
For this purpose, one can use so-called digital signatures.
An elegant scheme for digital signatures is to employ one-way functions, which are one of the most important concepts for
public-key cryptography. 
A crucial property of public-key cryptography based on one-way functions is that it provides a computationally simple algorithm for legitimate users (e.g., for key distribution or signing a document), 
whereas the problem for malicious agents is extremely computationally expensive. 
It should be noted that the very existence of one-way functions is still an open conjecture.
Thus, the security of corresponding public-key cryptography tools is based on unproven assumptions about the computational facilities of malicious parties. 

Assumptions on the security status of cryptographic tools may change with time. 
For example, breaking the RSA cryptographic scheme is  at least as hard as factoring large integers~\cite{RSA1978}.
This task is believed to be extremely hard for classical computers, but it appeared to be solved in polynomial time with the use of a large-scale quantum computer using Shor's algorithm~\cite{Shor1997}.  
A full-scale quantum computer that is capable of launching Shor's algorithm for realistic RSA key sizes in a reasonable time is not yet created.
At the same time, there are no identified fundamental obstacles that prevent from development of quantum computers of a required scale.
Thus, prudent risk management requires defending against the possibility that attacks with quantum computers will be successful.

A solution for the threat of creating quantum computers is the development of a new type of cryptographic tools that strive to remain secure even under the assumption that the malicious agent has a large-scale quantum computer.
This class of quantum-safe tools consists of two distinct methods~\cite{Wallden2019}.
The first is to replace public-key cryptography with quantum key distribution, which is a hardware solution based on transmitting information using individual quantum objects. 
The main advantage of this approach is that the security relies not on any computational assumptions, but on the laws of quantum physics~\cite{Gisin2002}. 
However, quantum key distribution technologies today face a number of important challenges such as secret key rate, distance, cost, and practical security~\cite{Lo2016}.

Another way to guarantee the security of communications is to use so-called post-quantum (also known as quantum-resistant) algorithms, 
which use a specific class of one-way functions that are believed to be hard to invert both using classical and quantum computers~\cite{Bernstein2017}. 
The main criticism of post-quantum cryptography is the fact that they are again based on computational assumptions so that there is no strict proof that they are long-term secure. 

In our work, we consider a scenario, where an adversary finds a way to violate basic mathematical assumptions underlying the security of a particular cryptographic primitive.
Thus, the adversary becomes able to perform successful attacks on information processing systems, which employ the vulnerable cryptographic primitive in their workflow.
At the same time, it is in the interests of the adversary that the particular cryptographic primitive be in use as long as possible since its replacement with  another  one eliminates an obtained advantage.
Thus, the preferable strategy of an attacker is to hide the fact that the underlying cryptographic primitive has been broken. 
It can be realized by performing attacks in such a way that their success could be explained by some other factors (e.g. user negligence, hardware faults, and etc.), but not the underlying cryptographic primitive.
An illustrative example of such a strategy is hiding the information about the successes of the Enigma system cryptoanalysis during World War II.

Broadly speaking, the question we address in the present work is as follows: Is it possible to supply a new generation of post-quantum cryptographic algorithms with some kind of alarm indicating that they are broken?
We argue that the answer to this question is partially positive, and the property, which we refer to as a \emph{forgery detection availability}, 
can be realized by properly designed hash-based signatures. 
The intuitive idea behind this property is that a forgery of a hash-based signature most likely results in finding a collision with respect to the underlying cryptographic hash function (see Fig.~\ref{fig:preimages}), and so the demonstration of this collision can serve as convincing evidence of the forgery and corresponding vulnerability of the employed cryptographic hash function.
We refer to the mathematical scheme for the evidence of the forgery event as a \emph{proof-of-forgery} concept.
We also would like to emphasize the fact that some of widespread hash functions have been compromised after their publication~\cite{MD4Collisions,MD5Collisions,SHA1Collisions}, therefore the considered problem is more than just of academic interest.

\begin{figure}[t]
	\centering
	\includegraphics[width=\linewidth]{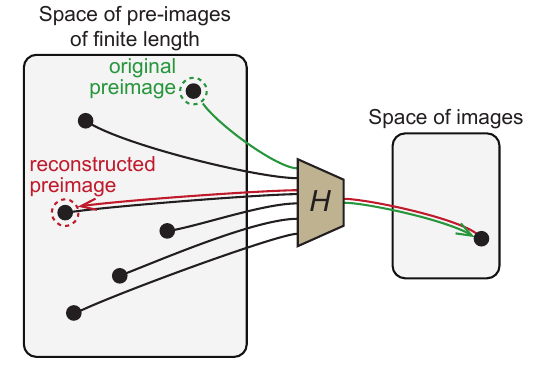}
	\caption{Demonstration of the idea behind proving the fact of the hash-based signature forgery. 
		In order to forge the signature, an adversary finds a valid preimage for a given image of a cryptographic hash function. 
		If the size of the preimages space is large enough then the preimage obtained adversary is most likely different from the legitimate user's one. 
		Disclosing the colliding preimage could serve as evidence that a particular hash function is vulnerable.}
	\label{fig:preimages}
\end{figure}

In the present work, we illustrate the forgery detection availability property for Lamport~\cite{PQBookChapter,Lamport1979} and Winternitz~\cite{PQBookChapter} one-time signatures schemes.
First, we consider the Lamport scheme, which is paradigmatically important: It is the first and the simplest algorithm among hash-based schemes.
However, the Lamport scheme is not widely used in practice. 
Then we analyze the Winternitz scheme, particularly the variant presented in Ref.~\cite{Hulsing2013}, which can be considered as a generalization of the Lamport scheme that introduced a size-performance trade-off.
Variations of the Winternitz scheme are used as building blocks in a number of modern hash-based signatures, 
such as LMS~\cite{LMS}, XMSS~\cite{XMSS}, SPHINCS~\cite{SPHINCS} and its improved modifications~\cite{SPHINCSplus,GravitySPHINCS}, as well as applications such as IOTA distributed ledger~\cite{IOTA}.

The paper is organized as follows.
In Sec.~\ref{sec:hash}, we give a short introduction to the scope of hash-based signatures.
In Sec.~\ref{sec:generalscheme}, we provide a general scheme of detecting signature forgery event and define a property of the $\varepsilon$-forgery detection availability ($\varepsilon$-FDA).
In Sec.~\ref{sec:Lamport}, we consider the $\varepsilon$-FDA property for the generalized Lamport one-time signature (\LOTS{}).
In Sec.~\ref{sec:Winternitz}, we consider the $\varepsilon$-FDA property for the Winternitz one-time signature (\WOTS{}).
We summarize the results of our work in Sec.~\ref{sec:Concl}. 

\section{\uppercase{Hash-based signatures}} \label{sec:hash}

Hash-based digital signatures~\cite{Bernstein2017,PQBook} are one of the post-quantum alternatives for currently deployed signature schemes, which have gained a significant deal of interest.
The attractiveness of hash-based signatures is mostly due to low requirements to construct a secure scheme. 
Typically, a cryptographic random or pseudorandom number generator is needed, and a function with some or all of preimage, 
second-preimage, and collision resistance properties, perhaps, in their multi-target variety~\cite{SPHINCSplus,XMSSTheory,PQBookChapter}. 
Some schemes for hash-based signatures require a random oracle assumption~\cite{random_oracle} to precisely compute their bit security level~\cite{LMSTheory}.

Up to date known quantum attacks based on Grover's algorithm~\cite{Grover1996} are capable to find a preimage and a collision with time growing sub-exponentially with a length of hash function output~\cite{BrassardSearching,BrassardCollisions}.
Specifically, it is proven that in the best-case scenario Grover's algorithm gives a quadratic speed-up in a search problem~\cite{Grover1996}.
While this area is a subject of ongoing research and debates~\cite{BernsteinSHARCS,BernsteinQuantMultiTargetCollisions,CNSQCollisions}, hash-based signatures are considered resilient against quantum computer attacks.
Meanwhile, the overall performance of hash-based digital signatures makes them suitable for the practical use, and several algorithms have been proposed for standardization
by NIST (SPHINCS$^+$ \cite{SPHINCSplus}, Gravity-SPHINCS \cite{GravitySPHINCS}) and IETF (LMS \cite{LMS}, XMSS \cite{XMSS}).

We note that the hash-based digital signature scheme can be instantiated with any suitable cryptographic hash function.
In practice, standardized hash functions, such as SHA, are used for this purpose since they are presumed to satisfy all the necessary requirements.
The availability of changing a core cryptographic primitive without a change in the functionality of the whole information security system fits a paradigm of \emph{crypto-agility},
which is the basic principle of modern security systems development with the built-in possibility of component replacement.

\section{\uppercase{Proving the fact of a forgery}} \label{sec:generalscheme}

Here we present a general framework for the investigation of the proof-of-forgery concept. 
We start our consideration by introducing a generic deterministic digital signature scheme.

\begin{mydef}[Deterministic digital signature scheme]
	A deterministic digital signature scheme (DDSS) $\DSS=(\kg, \sign, \vf)$ is a triple of algorithms that allows performing the following tasks:
	\begin{itemize}
		\item $\DSS.\kg(1^n)\rightarrow (\sk,\pk)$ is a probabilistic key generation algorithm that outputs a secret key \sk, aimed at signing messages, and a public key \pk, aimed at checking signatures validity, on input of a security parameter $1^n$.
		\item $\DSS.\sign(\sk,M)\rightarrow \sigma$ is a \emph{deterministic} algorithm that outputs a signature $\sigma$ under secret key $\sk$ for a message $M$.
		\item $\DSS.\vf(\pk,\sigma,M)\rightarrow v$ is a verification algorithm that outputs $v=1$ if the signature $\sigma$ of the signed message $M$ is correct under the public key $\pk$, and $v=0$ otherwise.
	\end{itemize}
\end{mydef}

We note that the deterministic property of the DDSS is defined by the fact that for a given pair $(\sk,M)$ the algorithm $\DSS.\sign(\sk,M)$ always generates the same output.

The standard security requirement for digital signature schemes is their existential unforgeability under chosen message attack (EU-CMA).
The chosen message attack setting allows the adversary to choose a set of messages that a legitimate user has to sign.
Then the existential unforgeability property means that the adversary should not be able to construct any valid message-signature pair $(M^\star, \sigma^\star)$, where the message $M^\star$ is not previously signed by a legitimate secret key holder. 
In the present work, we limit ourselves to the case of one-time signatures, so the adversary is allowed to obtain a signature for a single message only.
The generalization to the many-time signature schemes is left for future research.

In the present work, we consider a stronger security requirement known as \emph{strong unforgeability} under chosen message attack (SU-CMA).
A DDSS is said to be SU-CMA if it is EU-CMA, and given signature $\sigma$ on some message $M$, the adversary cannot even produce a new signature $\sigma^{*}\neq\sigma$ on the message $M$.
We note that SU-CMA schemes are used for constructing chosen-ciphertext secure systems and group signatures~\cite{Boneh2006,StrongUnforg}.
The security is usually proven under the assumption that the adversary is not able to solve some classes of mathematical problems, such as integer factorization, discrete logarithm problem, or inverting a cryptographic hash function.
Here we consider the case where this assumption is not fulfilled.

Let us discuss the following scenario involving three parties: An honest legitimate signer $\mathcal{S}$, an honest receiver $\mathcal{R}$, and an adversary $\mathcal{A}$.
At the beginning ({\bf step~0}) we assume that $\mathcal{S}$ possesses  a pair $(\sk,\pk)\leftarrow S.\kg$, while $\mathcal{R}$ and $\mathcal{A}$ have a public key $\pk$ of $\mathcal{S}$, 
and they have no any information about the corresponding secret key $\sk$ (see Table \ref{tab:steps}).

\begin{table*}[t] \scriptsize
	\begin{center}
		\begin{tabular}{|c|l|l|l|}
			\hline
			&  Signer $\mathcal{S}$ & Adversary $\mathcal{A}$ & Receiver $\mathcal{R}$ \\ \hline 
			Step 0 & $\sk$, $\pk$ &  $\pk$ & $\pk$ \\
			Step 1 & $\sk$, $\pk$, $(M, \sigma)$  & $\pk$, $(M, \sigma)$  & $\pk$, [$(M, \sigma)$]  \\
			Step 2 & $\sk$, $\pk$, $(M, \sigma)$  & $\pk$, $(M, \sigma)$, $(M^\star, \sigma^\star)$  & $\pk$, [$(M, \sigma)$], $(M^\star, \sigma^\star)$  \\
			Step 3 & $\sk$, $\pk$, $(M, \sigma)$, $(M^\star, \sigma^\star)$ & $\pk$, $(M, \sigma)$, $(M^\star, \sigma^\star)$  &  $\pk$, [$(M, \sigma)$], $(M^\star, \sigma^\star)$ \\
			Step 4 & $\sk$, $\pk$, $(M, \sigma)$, $(M^\star, \sigma^\star)$, $E$ & $\pk$, $(M, \sigma)$, $(M^\star, \sigma^\star)$, [$E$] & $\pk$,  [$(M, \sigma)$], $(M^\star, \sigma^\star)$, $E$ \\
			\hline
		\end{tabular}
		\vspace{4pt}
		\caption{Message-signature pairs and keys available to involved parties on each step of the scenario, 
			where the adversary $\mathcal{A}$ makes a successful CMA obtaining a signature $\sigma$ for some message $M$, 
			and forges a signature $\sigma^\star$ for some new message $M^\star\ne M$ under the public key $\pk$ of the signer $\mathcal{S}$.
			However, the signer $\mathcal{S}$ is able to construct the corresponding proof-of-forgery message $E$ in order to convince the receiver $\mathcal{R}$ that the forgery event happened.
			Square brackets correspond to the optional message-signature transmission.}
		\label{tab:steps}
	\end{center}
\end{table*}

At {\bf step~1} $\mathcal{A}$ forces $\mathcal{S}$ to sign a message $M$ of $\mathcal{A}$'s choice.
In the result, $\mathcal{A}$ obtains a valid message-signature pair $(M,\sigma)$.
At this step, $\mathcal{R}$ may or may not know about the fact of signing $M$ by $\mathcal{S}$.

Then at {\bf step 2} $\mathcal{A}$ performs an existential forgery by producing a new message-signature pair $(M^\star,\sigma^\star)$ with $M^\star\neq M$. 
Below we introduce a formal definition of the signature forgery and specify two different cases.

\begin{mydef}[Signature forgery and its types]
	A signature $\sigma^\star$ is called a forged signature of the message $M^\star$ under the public key $\pk$ and the signature scheme $\DSS$ if $\DSS.\vf(\pk,\sigma^\star,M^\star) \rightarrow 1$,
	where the message $M^\star$ has not been signed by the legitimate sender possessing secret key $\sk$ corresponding to $\pk$.
	The following two cases are possible.
	\begin{itemize}
		\item A pair $(M^\star,\sigma^\star)$ is called a forgery of type I if the signature $\sigma^\star$ has been previously generated by the legitimate user a signature for some message other than $M^\star$. 
		That is, there is a message $M$ with $\DSS.\sign(\sk,M)\rightarrow\sigma^\star$ previously signed by a legitimate user.
		\item A pair $(M^\star,\sigma^\star)$ is called a forgery of type II if the signature $\sigma^\star$ has not been previously generated by the legitimate user. 
		That is, there has not been a message with signature $\sigma^\star$, signed by the legitimate user.
	\end{itemize}
\end{mydef}

The type I forgery can take place if the signature algorithm $\DSS.\sign$ calculates a \emph{digest} of an input message and then computes a signature of the corresponding digest. 
In this case, the adversary $\mathcal{A}$ may find a collision of the digest function, and then force the legitimate user to sign a first colliding message by using it as $M$, 
and automatically obtain a valid signature for the second colliding message (use it as $M^\star$). 

An example of type II forgery is the reconstruction of the $\sk$ from $\pk$ using an efficient algorithm (in analogy to the use of Shor's algorithm on a quantum computer for the RSA scheme). 
We note that in our consideration it is assumed that the only way for the adversary $\mathcal{A}$ to forge the signature for $M^\star$ is to employ advanced mathematical algorithms and/or unexpectedly powerful computational resources.
In other words, we do not consider any side-channel attacks or other forms of secret key ``stealing'', such as social engineering and others.

Coming back to the considered scenario, at the {\bf step 2}, $\mathcal{A}$ sends a pair $(M^\star, \sigma^\star)$ to $\mathcal{R}$ claiming that $M^\star$ was originally signed by $\mathcal{S}$.
If the signature is successfully forged by the adversary $\mathcal{A}$, then this could be the end of the story.

However, we suggest accomplishing this scenario by the following next steps.
At the {\bf step 3}, $\mathcal{R}$ sends a message $(M^\star, \sigma^\star)$ directly to $\mathcal{S}$ in order to request an additional confirmation.
Then $\mathcal{S}$ observes a valid signature $\sigma^\star$ of the corresponding message $M^\star$, which was not generated by him. 
The concrete issue we address in the present work is whether $\mathcal{S}$ is able to \emph{prove} the fact of a forgery event.
Here we formally introduce a proof-of-forgery concept, which is mathematical evidence that someone cheats with signatures by employing computational resources or advanced mathematical algorithms.

\begin{mydef}[Proof-of-forgery of type I]
	\label{def:PoFI}
	A set $E=(\pk, \sigma^\star, M, M^\star)$ is called a proof-of-forgery of type I (PoF-I) for a DDSS $\DSS$ if 
	for $M\neq M^\star$ there is a valid signature $\sigma^\star$ for these two messages, i.e. the following relations hold:
	\begin{equation}
	\DSS.\vf(\pk,\sigma^\star, M^\star) \rightarrow 1, \quad \DSS.\vf(\pk, \sigma^\star, M) \rightarrow 1.
	\end{equation}
\end{mydef}

Obviously, if the adversary $\mathcal{A}$ performs the type I forgery, then $\mathcal{S}$ is able to prove this fact by demonstrating $M$ to $\mathcal{R}$ at {\bf step 4}.
Thus, $\mathcal{S}$ and $\mathcal{R}$ have the complete PoF-I set $E=(\pk, \sigma^\star, M, M^\star)$, and they are sure that someone has an ability to break SUF-CMA property~\cite{Brendel2020},
which is typically beyond the consideration in standard computational hardness assumptions.
Moreover, they can use the set $E$ to prove the fact of the forgery event to any third party since $E$ contains a public key $\pk$.
We also note that it is possible to prove the fact of a forgery of type I for any DDSS.  
The situation in the PoF-II case is more complicated.

\begin{mydef}[Proof-of-forgery of type II]
	\label{def:PoFII}
	A set $E=(\pk,\widetilde{\sigma}^\star, \sigma^\star, M^\star)$ is called a proof-of-forgery of type II (PoF-II) for a DDSS $\DSS$ if 
	for a message $M^\star$ there are distinct valid signatures $\widetilde{\sigma}^\star\neq\sigma^\star$, i.e. the following relations hold:
	$$
	\DSS.\vf(\pk,\widetilde{\sigma}^\star, M^\star) \rightarrow 1, \quad \DSS.\vf(\pk,\sigma^\star, M^\star) \rightarrow 1.
	$$
\end{mydef}

The ability of the adversary $\mathcal{A}$ to perform a forgery of type II depends on a particular deterministic signature scheme $\DSS$.
Suppose that $\mathcal{A}$ has succeeded in obtaining $\sk$ from $\pk$ (e.g. by using Shor's algorithm and RSA-like scheme), 
then it is impossible for $\mathcal{S}$ to convince $\mathcal{R}$ that $(M^\star, \sigma^\star)$ was not generated by $\mathcal{S}$.
However, if the adversary $\mathcal{A}$ has succeeded in obtaining a valid, but \emph{different} secret key $\sk'\neq\sk$, 
then the legitimate sender $\mathcal{S}$ is able to construct the corresponding PoF-II set by calculating $\sign(\sk,M^\star)\rightarrow\widetilde{\sigma}^\star$ with $\widetilde{\sigma}^\star\neq\sigma^\star$.

As we show below this scenario is the case for properly designed hash-based signatures.
We consider particular examples of Lamport and Winternitz one-time signatures schemes. 
We show that under favourable circumstances $\mathcal{S}$'s signature $\widetilde{\sigma}^\star$ of the corresponding message $M^\star$ is different from  $\mathcal{A}$'s signature $\sigma^\star$, 
and $\mathcal{S}$ can send it as part of PoF-II to $\mathcal{R}$ at {\bf step 4}.
Thus, the PoF-II set is successfully constructed, so legitimate parties are aware of the break of the used DDSS.

Here, we introduce a definition of an adversary who successfully forged, which allows proving their forgery.
\begin{mydef}[$\varepsilon$-forgery detection availability]
	$\varepsilon$-forgery detection availability ($\varepsilon$-FDA) for a one-time DDSS $\DSS$ is defined by the following experiment.\\
	{\bf Experiment} ${{\rm Exp}}^{FDA}_{\DSS, n}(\mathcal{A})$\\
	\indent $(\sk,\pk)\leftarrow S.\kg(1^n)$\\
	\indent $(M^\star,\sigma^\star)\leftarrow \mathcal{A}^{\sign(\sk,\cdot)}$\\
	\indent Let $(M,\sigma)$ be the query-answer pair of $\sign(\sk, \cdot)$.\\
	\indent Return 1 iff $S.\sign(\sk,M^\star)\rightarrow\sigma^\star$, $S.\vf(\pk,\sigma^\star,M^\star)\rightarrow1$, and $M^\star\neq M$.\\
	Then the DSS scheme $\DSS$ has $\varepsilon$-FDA if there is no adversary $\mathcal{A}$ that succeeds with probability $\geq\varepsilon$.
\end{mydef}

\begin{myremark}
	In our consideration, we implicitly assume that the parties are able to communicate with each other via authentic channels, e.g. when $\mathcal{R}$ sends a request to $\mathcal{S}$ at {\bf step 3}.
	One can see that in order to enable the detection of the forgery event, the authenticity of the channel should be provided with some different primitives rather than employed signatures.
	For example, one can use message authentication codes (MACs), which can be based on information-theoretical secure algorithms and symmetric keys. 
\end{myremark}
\begin{myremark}
	In the considered scheme honest users only become aware of the fact of forgery event. 
	However, the scheme does not allow determining who exactly in this scenario has such powerful computational capabilities.
	Indeed, $\mathcal{S}$ is not sure whether the signature $\sigma^\star$ is forged by $\mathcal{R}$ or by $\mathcal{A}$.
	That is why it is advisable for $\mathcal{S}$ also to send evidence $E$ to $\mathcal{A}$ as well. 
	At the same time, $\mathcal{R}$ is not sure, who is the original author of $\sigma^\star$.
	It is a possible case that $\mathcal{S}$ has forged its own signature (say, obtained two messages $M$ and $M^\star$ with a same signature $\sigma=\sigma^\star$), and sent a message $M$ to $\mathcal{A}$, who then just forwarded it $\mathcal{R}$.
	It may be in the interest of a malicious $\mathcal{S}$ to reveal $M^\star$ at the right moment and claim that it was a forgery.
\end{myremark}

\section{\uppercase{$\varepsilon$-FDA for Lamport signatures}\label{sec:Lamport}}
Here we start with a description of a generalized Lamport single bit one-time DDSS.
Consider a cryptographic hash function $H:\{0,1\}^*\rightarrow\{0,1\}^n$.
The $(n,\delta)$-Lamport one-time signature ($(n,\delta)$-\LOTS) scheme for single bit message $M\in\{0,1\}$ has the following construction.

\noindent \emph{Key pair generation algorithm} ($(\sk,\pk)\leftarrow(n,\delta)$-\LOTS.\kg).
	The algorithm generates secret and public keys in a form $(\sk_0, \sk_1)$ and $(\pk_0, \pk_1)$, with
	$\sk_i\stackrel{\$}{\leftarrow}\{0,1\}^{n+\delta}$ and $\pk_i := H(\sk_i)$ (see Fig.~\ref{fig:Lamprot}).
	Here and after $\stackrel{\$}{\leftarrow}$ stands for uniformly random sampling from a given set.
	
\noindent \emph{Signature algorithm} ($\sigma\leftarrow(n,\delta)$-\LOTS.\sign$(\sk,M))$.
	The algorithm outputs half of the secret key as a signature: $\sigma := sk_M$.
	
\noindent \emph{Verification algorithm} ($v\leftarrow(n,\delta)$-\LOTS.$\vf(\pk,\sigma,M))$).
	The algorithm outputs $v:=1$, if $H(s)=\pk_M$, and its output is $0$ otherwise.

\begin{figure}
	\centering
	\includegraphics[width=0.8\linewidth]{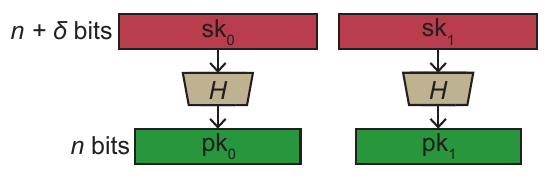}
	\caption{Basic principle of the public key construction in the $(n,\delta)$-\LOTS\,scheme.}
	\label{fig:Lamprot}
\end{figure}

The security of the $(n,\delta)$-\LOTS~scheme is based on the fact that in order to forge a signature for a bit $M$ it is required to invert the used one-way function $H$ for a part of the public key $\pk_M$, 
that is traditionally assumed to be computationally infeasible.

In our work, we particularly stress the importance of inequality between space sizes of secret keys and public keys.
Specifically, we demonstrate that for sufficiently large $\delta$ even if an adversary finds a correct preimage $\sk_M^\star$, such that $H(\sk_M^\star)=\pk_M$,
the obtained value is different from the original $\sk_M$ used for calculating $\pk_M$ by the legitimate user.
Then the signature of an honest user is different from a forged signature, and so the forgery event can be revealed.

Before turning to the main theorem, we prove the following Lemma.

\begin{mylemma}\label{lemma1}
	Consider a function $f: \{0,1\}^{n+\delta}\rightarrow\{0,1\}^n$ with $n\gg1$ and $\delta\geq0$ taken at random from the set all functions from $\{0,1\}^{n+\delta}$ to $\{0,1\}^n$. 
	Let $y_0=f(x_0)$ for $x_0$ taken uniformly at random from  $\{0,1\}^{n+\delta}$.
	Define a set 
	\begin{equation}
	{\rm Inv}(y_0):=\{x\in\{0,1\}^{n+\delta}| f(x)=y_0\}
	\end{equation}	
	of all preimages of $y_0$ under $f$. 
	Consider a randomly taken preimage $X\stackrel{\$}{\leftarrow}{\rm Inv}(y_0)$.
	Then the probability to obtain the original preimage $X=x_0$ has the following lower and upper bounds:
	\begin{itemize}
		\item[a)] $\Pr(X=x_0)> \exp(-2^{\delta})$;
		\item[b)] $\Pr(X=x_0) < 5.22\times2^{-\delta}$.  
	\end{itemize}
\end{mylemma}
\begin{proof}
	Let $\mathcal{N}:=|{\rm Inv}(y_0)|$ be a number of preimages of $y_0$ under $f$.	
	Due to the random choice of $f$, it is given by $\mathcal{N}=1+\widehat{\mathcal{N}}$, 
	where $\widehat{\mathcal{N}}$ is a random variable having binomial distribution ${\rm Bin}(2^{-n},2^{n+\delta}-1)$ with the success probability $2^{-n}$ and number of trials equal to $2^{n+\delta}-1$.
	Then the corresponding probability that a randomly chosen element $X$ from ${\rm Inv}(y_0)$ is equal to $x_0$ is as follows: 
	\begin{equation}\label{eq:summation}
	\Pr(X=x_0)=\sum_{N=1}^{2^{n+\delta}-1}\frac{1}{N}\Pr(\mathcal{N}=N).
	\end{equation}
	In order to obtain the lower bound for $\Pr(X=x_0)$, we consider only the first term in Eq.~\eqref{eq:summation} and arrive at the following inequality: 
	\begin{multline}
	    \Pr(X=x_0)>\Pr(\mathcal{N}=1)=\left(1-2^{-n}\right)^{2^{n+\delta}-1}\\
	    \simeq\left(1-2^{-n}\right)^{2^{n+\delta}}\simeq\exp(-2^\delta),
	\end{multline}
	where we use the fact that $(1-2^{-n})^{2^n}\simeq \exp(-1)$ for $n\gg 1$.
	This proves part a) of Lemma~\ref{lemma1}. 
	
	In order to obtain the upper bound for $\Pr(X=x_0)$, we split the sum in Eq.~\eqref{eq:summation} into following two parts:  
	\begin{multline}
	\Pr(X=x_0)=\sum_{N=1}^{N_0}\frac{1}{N}\Pr(\mathcal{N}=N)+\\\sum_{N=N_0+1}^{2^{m}-1}\frac{1}{N}\Pr(\mathcal{N}=N)
	\end{multline}
	where $N_0:=k2^\delta\geq 1$ for some $k\in(0,1)$.
	The first part can be bounded as follows:
	\begin{equation} \label{eq:bnd1}
	\begin{split}
	\sum_{N=1}^{N_0}\frac{1}{N}\Pr(\mathcal{N}=N)&\leq\Pr(\mathcal{N}\leq N_0)\\
	&\leq\frac{(2^{n+\delta}-N_0)2^{-n}}{(2^{\delta}-N_0)^2}\\&<\frac{2^{-\delta}}{(k-1)^2},
	\end{split}
	\end{equation}
	where we use a bound for the cumulative binomial distribution function~\cite{Feller}.
	For the second part we consider the following bound:
	\begin{equation} \label{eq:bnd2}
	\begin{split}
	\sum_{N=N_0+1}^{2^{m}-1}\frac{1}{N}\Pr(\mathcal{N}=N)&<\frac{1}{N_0}\sum_{N=N_0+1}^{2^{m}-1}\Pr(\mathcal{N}=N)\\
	&<\frac{1}{N_0}=\frac{2^{-\delta}}{k}.
	\end{split}
	\end{equation}
	By combining Eq.~\eqref{eq:bnd1} with Eq.~\eqref{eq:bnd2} and setting $k:=0.36$, which corresponds to a minimum of $\Pr(X=x_0)$, we obtain $\Pr(X=x_0)<5.22\times2^{-\delta}$.
	This proves part b) of Lemma~\ref{lemma1}. 
\end{proof}

\begin{myremark}
	The bound for commutative binomial distribution employed in~\eqref{eq:bnd1} is rather rough, however, it is quite convenient for the purposes of further discussion.
	Tighter bound can be obtained, e.g. using the technique from Ref.~\cite{Zubkov2013}.
\end{myremark}

Next we assume that the number of $(n+\delta)$-bit preimages for $H(x)$ with $x \stackrel{\$}{\leftarrow} \{0,1\}^{n+\delta}$ behaves in the same way as for a random function from $\{0,1\}^{n+\delta}$ to $\{0,1\}^n$.
The main result on the FDA property of the $(n,\delta)$-\LOTS{} scheme can be formulated as follows:
\begin{mytheorem}\label{LOTSTheorem}
	$(n,\delta)$-\LOTS~scheme has $\varepsilon$-FDA with $\varepsilon<5.22\times2^{-\delta}$. 
\end{mytheorem}

\begin{proof}
	Consider an adversary who successfully forged a signature $\sigma^\star$ for message $M^\star$.
	From the construction of the signature scheme we have $H(\sigma^\star)=\pk_{M^\star}$.
	According to part b) of Lemma~\ref{lemma1}, the probability that the obtained value $\sigma^\star$ coincides with the original value $\sk_{M^\star}$ is bounded by $5.22\times2^{-\delta}$.
	It follows from the fact that $\sk_{M^\star}$ is generated uniformly randomly from the set $\{0,1\}^{n+\delta}$.
	Therefore, with  probability at least $1-5.22\times2^{-\delta}$ 
	the legitimate user's signature $\sk_{M^\star}\leftarrow(n,\delta)$-\LOTS$.\sign(\sk,M^\star)$ is different from the adversary's signature $\sigma^\star$ and the presence of the forgery event is then proven.
\end{proof}

\begin{myremark}
	We note that it is extremely important to employ true randomness in the $\kg$ algorithm in order to provide the independence between the results of the adversary and the original value of $\sk$.
	For this purpose, one can use, for example, certified quantum random number generators. 
\end{myremark}
\begin{myremark}
	We see that $\varepsilon$-FDA property appears only for high enough values of $\delta$.
	Meanwhile it follows from part a) of  Lemma~\ref{lemma1} that for a common case of $\delta=0$ the probability for the adversary to obtain the original value $\sk_1$ is at least $\exp(-1)\approx 0.368$ that is non-negligible.
\end{myremark}

\section{\uppercase{$\varepsilon$-FDA for the Winternitz Signature Scheme}}\label{sec:Winternitz}

Here we consider an extension of L-OTS scheme which allows signing messages of $L$-bit length.
We base our approach on a generalization of the Winternitz one-time signature (W-OTS) scheme presented in Ref.~\cite{Hulsing2013}, known as \WOTS, and used in XMSS~\cite{XMSSTheory}, SPHINCS~\cite{SPHINCS} and SPHINCS$^+$~\cite{SPHINCSplus}.
We note that \WOTS~is SU-CMA scheme~\cite{Hulsing2013}.
We refer our scheme $(n,\delta, L, \nu)$-{\WOTS} and construct it as follows.

Let us introduce the parameter $\nu\in\{1,2,\ldots\}$ defining blocks length in which a message is split during a signing algorithm, where we assume that $L$ is a multiple of $\nu$.
Let us introduce the following auxiliary values:
$w:=2^\nu, \quad 
l_1 := \lceil L/\nu \rceil,
l_2 := \lfloor\log_2(l_1(w-1))/\nu\rfloor+1, \quad
l :=l_1+l_2.$
Then we consider a family of one-way functions:
\begin{equation}
    f_{\bf r}^{(i)}: \{0,1\}^{n+\delta(w-i)} \rightarrow \{0,1\}^{n+\delta(w-i-1)},
\end{equation}
where $i\in\{1,\ldots, w-1\}$ and a parameter ${\bf r}$ belongs to some domain $\mathcal{D}$.
The employ of this parameter can correspond to XORing the result of some hash function family with a random bit-mask, as it considered in Ref.~\cite{Hulsing2013}).
We assume that $f_{\bf r}^{(i)}$ satisfies the random oracle assumption for a uniformly randomly chosen ${\bf r}$ from $\mathcal{D}$.

We then introduce a chain function $F_{\bf r}^{(i)}$, which we define recursively in the following way: 
\begin{equation}
\begin{split}
    F_{\bf r}^{(0)}(x)&=x,\\
    F_{\bf r}^{(i)}(x)&=f^{(i)}_{\bf r}(F_{\bf r}^{(i-1)}(x))~\text{for}~i\in\{1,\ldots, w-1\}.
\end{split}
\end{equation}

The algorithms of $(n,\delta, L, \nu)$-{\WOTS}~scheme are the following:

\noindent \emph{Key pair generation algorithm} ($(\sk,\pk)\leftarrow(n,\delta, L, \nu)$-{\WOTS}.\kg).
First the algorithm generates a secret key in the following form: 
$\sk:=({\bf r}, \sk_1,\sk_2,\ldots,\sk_l), \text{ with } \sk_i \stackrel{\$}{\leftarrow} \{0,1\}^{n+\delta(w-1)} \text{ and } {\bf r}\stackrel{\$}{\leftarrow}\mathcal{D}$ (see Fig.~\ref{fig:WOTS}).
Then a public key composed of the randomizing parameter ${\bf r}$ and results of the chain function employed to $\sk_i$ as follows:	$\pk:=({\bf r}, \pk_1,\pk_2,\ldots,\pk_l)$ with $\pk_i:=F_{\bf r}^{(w-1)}(\sk_i)$.

\begin{figure}
	\centering
	\includegraphics[width=\linewidth]{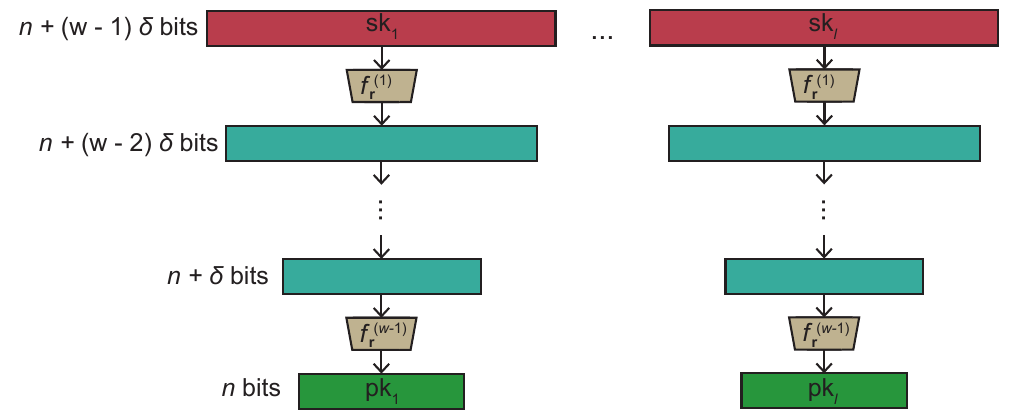}
	\caption{Basic principle of the public key construction in the $(n,\delta, L, \nu)$-{\WOTS} scheme.}	
	\label{fig:WOTS}
\end{figure}
	
\noindent \emph{Signature algorithm} ($\sigma\leftarrow(n,\delta, L, \nu)$-{\WOTS}$.\sign(\sk,M))$.
First the algorithm computes base $w$ representation of $M$ by splitting it into $\nu$-bit blocks ($M=(m_1,\ldots,m_{l_1})$, where $m_i\in\{0,\ldots,w-1\}$).
We call it a message part.
Then the algorithm computes a checksum 
$C:=\sum_{i=1}^{l_1}(w-1-m_i)$ and its base $w$ representation $C=(c_1,\ldots,c_{l_2})$. 
We call it a checksum part.
Define an extended string $B=(b_1,\ldots b_l):=M\|C$ as the concatenation of message and checksum parts.
Finally, the signature is generated as follows: $\sigma:=(\sigma_1,\sigma_2,\ldots,\sigma_l)$ with $\sigma_i:=F_{\bf r}^{(b_i)}(\sk_i).$
	
\noindent \emph{Verification algorithm} ($v\leftarrow(n,\delta, L, \nu)$-{\WOTS}.$\vf(\pk,\sigma,M))$).
The idea of the algorithm is to reconstruct a public key from a given signature $\sigma$ and then to check whether it coincides with the original public key $\pk$.
First, the algorithm computes a base $w$ string $B=(B_1,\ldots, B_l)$ in the same way as in the signature algorithm (see above). 
Then for each part of the signature $\sigma_i$ the algorithm computes the remaining part of the chain as follows: 
\begin{equation}
	\pk_i^{\rm check}:=f^{(w-1)}_{\bf r}\circ\ldots\circ f^{(b_i+1)}_{\bf r}(\sigma_i),
\end{equation} 
where $\circ$ stands for the standard functions composition.
If $\pk_i^{\rm check}=\pk_i$ for all $i\in\{1,\ldots,l\}$, then the algorithm outputs $v:=1$, otherwise $v:=0$.

The main result on the FDA property of the $(n,\delta, L, \nu)$-{\WOTS}~scheme can be formulated as follows:

\begin{mytheorem}\label{WOTSTheorem}
	The $(n,\delta, L, \nu)$-{\WOTS}~scheme has the $\epsilon$-FDA property with $\epsilon<5.22\times2^{-\delta}$.
\end{mytheorem}

\begin{proof}
	
	Consider a scenario of successful CMA on the $(n,\delta, L, \nu)$-{\WOTS}~scheme, in which an adversary first forces a legitimate user with public key 
	$\pk=({\bf r},\pk_1,\ldots,\pk_l)$ to provide him a signature $\sigma=(\sigma_1,\ldots,\sigma_l)$ for some message $M$, and then generate a valid signature $\sigma^\star=(\sigma_1^\star,\ldots,\sigma_l^\star)$ for some message $M^\star\neq M$.
	Let $(m_1,\ldots,m_{l_1})$ and $(m^\star_1,\ldots,m^\star_{l_1})$ be the $w$-base representations of $M$ and $M^\star$ correspondingly.
	Consider extended $w$-base strings $B=(b_1^0,\ldots,b_l^0)$ and $B^\star=(b_1^\star,\ldots,b_l^\star)$ generated by adding checksum parts.
	It easy to see that for any distinct $M$ and $M^\star$ there exists at least one position $j\in\{1,\ldots,l\}$ such that $b^\star_{j}<b_{j}$.
	Indeed, even if for all positions $i\in\{1,\ldots,l_1\}$ it happened that $m^\star_i>m_i$, from the definition of checksum it follows that there exists a position $j\in\{l_1+1,\ldots,l_2\}$ in checksum parts such that $b^\star_{j}<b_{j}$.
	
	Since $\sigma^\star$ is a valid signature for $M^\star$ we have
	\begin{equation}~\label{eq:trueforg}
	f_{\bf r}^{(w-1)}\circ\ldots\circ f_{\bf r}^{(b^\star_{j}+1)}(\sigma^\star_{j})
	=\pk_{j}.
	\end{equation}
	One can see that forgery event will be detected if the $j^{\rm th}$ part of the legitimate user's signature of $M^\star$ is different from the forged one (see also Fig.~\ref{fig:forg}), so that:
	\begin{equation}
	\widetilde{\sigma}^\star_{j}:=F_{\bf r}^{(b_j^\star)} (\sk_j)\neq \sigma^\star_{j}.
	\end{equation}
	
	\begin{figure}
		\centering
		\includegraphics[width=0.6\linewidth]{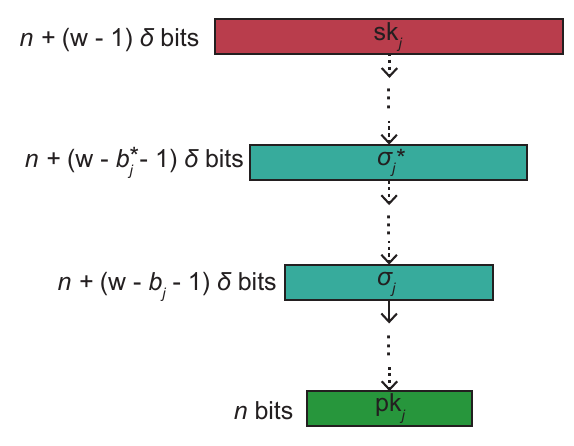}
		\caption{Illustration of the principle of PoF-II construction for the $(n,\delta, L, \nu)$-{\WOTS}~scheme.}
		\label{fig:forg}
	\end{figure}
	
	Consider two possible cases.
	The first is that the condition~\eqref{eq:trueforg} is fulfilled, but the following relation holds true:
	\begin{equation}
	f_{\bf r}^{(b_{j})}\circ\ldots\circ f_{\bf r}^{(b^\star_{j}+1)}(\sigma^\star_{j})\neq \sigma_j.
	\end{equation}
	
	In this case we obtain $\widetilde{\sigma}^\star_{j}\neq\sigma^\star_{j}$ with unit probability since 
	\begin{equation}
	\sigma_j=f_{\bf r}^{(b_{j})}\circ\ldots\circ f_{\bf r}^{(b^\star_{j}+1)}(\widetilde{\sigma}^\star_{j})\neq f_{\bf r}^{(b_{j})}\circ\ldots\circ f_{\bf r}^{(b^\star_{j}+1)}(\sigma^\star_{j}).
	\end{equation}
	In the second case we have the following identity:
	\begin{equation}
	f_{\bf r}^{(b_{j})}\circ\ldots\circ f_{\bf r}^{(b^\star_{j}+1)}(\sigma^\star_{j})=\sigma_j,
	\end{equation}
	which automatically implies the fulfilment of Eq.~\eqref{eq:trueforg}.
	Consider the function
	\begin{equation}
	    F:=f_{\bf r}^{(b_{j})}\circ\ldots\circ f_{\bf r}^{(b^\star_{j}+1)}: 
	\{0,1\}^{n^*+\delta\,\Delta}\rightarrow\{0,1\}^{n^*},
	\end{equation}
	where $\Delta:=b^0_{j}-b^\star_{j}\geq 1$ and $n^*:=n+\delta(w-b^\star_{j}-1)$.
	This function satisfies random oracle assumptions, since each of $\{f_{\bf r}^{(k)}\}_{k=b_j^\star}^{b_j}$ do.
	The random oracle assumption implies that the number of preimages of $F$ is the distributes in the same way as for random random function.
	Moreover, due to the fact that the values of $\sk_i$ are generated uniformly at random, the inputs for $F$ are also uniformly distributed over the corresponding domains.
	So, according to part b) of Lemma~\ref{lemma1}, we have the probability of the adversary to obtain $\sigma_j^\star=\widetilde{\sigma}_j^\star$ is bounded by $\epsilon < 5.22\times 2^{-\delta\,\Delta}\leq 5.22\times 2^{-\delta}$.
\end{proof}

\begin{myremark}
	One can see that the excess of the preimage space size over the image space size, given by $\delta$, is a crucial condition for the FDA property both for \LOTS{} and \WOTS{} schemes.
	We also note that in the case of the \WOTS{} scheme it is important to have such excess for all the elements of the employed one-way chain.
\end{myremark}

\begin{myremark}
	One can imagine a situation where an attacker is able to compute ${\rm sk}_i^\star \neq {\rm sk}_i$ for some $i\in\{1,\ldots,l\}$, such that
$F^{(w-1)}({\rm sk}_i^\star) = {\rm pk}_i$.
	This means that the chain collides somewhere on the way to ${\rm pk}_i$.
	As long as an attacker reveals forgeries only above this collision, there is no way for the legitimate user to present a different signature on the same message. 
	However, in line with our proof, the probability of such a scenario is limited by $\epsilon$.
\end{myremark}

In $WOTS^+$ scheme consider an
attacker who is able to compute some value $sk_i' \neq sk_i$, such that
$f^{w-1}(sk_i') = pk_i$ for one of the public-key chunks $pk_i$, this
means that the chain collides somewhere on the way to $pk_i$. One can notice 
that as long as an attacker reveals forgeries only above this collision, there is no way for the legitimate user to present a different signature on the same
message. But following our proof, this situation is limited by $\epsilon$.

\section{\uppercase{Conclusion}}\label{sec:Concl}

In this work, we have considered the $\varepsilon$-FDA property of DDSS that allows detecting a forgery event generated by advanced mathematical algorithms and/or unexpectedly powerful computational resources.
We have shown that this property is fulfilled for properly-designed hash-based signatures, in particular, for \LOTS{} and \WOTS{} schemes with properly tuned parameters.
As we have noted, the probability of the successful demonstration of the DDSS forgery event depends on an excess of preimage space sizes over image space sizes and using true randomness in the generation of secret keys in hash-based DDSS.
The important next step is to study this property for other types of hash-based signatures.

Our observation is important in the view of the crypto-agility paradigm.
Indeed, the considered forgery detection serves as an alarm that the employed cryptographic hash function has a critical vulnerability and it has to be replaced.
We note that a similar concept has been recently considered for detecting brute-force attacks on cryptocurrency wallets in the Bitcoin network~\cite{Kiktenko2019}. 
Namely, it was considered the alarm system that detects the case of stealing coins by finding a secret-public key pair for standard elliptic curve digital signature algorithm (ECDSA) used in the Bitcoin system, such that a public key hash of adversary equals a public key hash of a legitimate user.
This kind of alarm system can be of particular importance in view of the development of quantum computing technologies~\cite{Fedorov2018}.

\section*{\uppercase{Acknowledgements}}

We thank A.I. Ovseevich, A.A. Koziy, E.K. Alekseev, L.R. Akhmetzyanova, and L.A. Sonina for fruitful discussions.
This work is partially supported by Russian Foundation for Basic Research (18-37-20033).

\bibliographystyle{apalike}
{\small\bibliography{bibliography}}

\end{document}